\documentclass[smallextended]{svjour3}

\usepackage{amsmath,amsfonts,braket,algorithm,graphicx}

\newcommand{\kb}[1]{\mathbf{\left[#1\right]}}

\newcommand{\PM}{{M}}
\newcommand{\PN}{{N}}

\newcommand{\trd}[1]{\left|\left| #1 \right| \right|}

\newcommand{\st}{\text{ } | \text{ }}

\newcommand{\Hmin}{H_\infty}
\newcommand{\Hextd}{\bar{H}}

\newcommand{\hw}{w}

\newcommand{\tdl}{\left|\left|}
\newcommand{\tdr}{\right|\right|}

\floatname{algorithm}{Protocol}

\begin{document}
\title{Quantum Sampling and Entropic Uncertainty}
\author{
Walter O. Krawec
}
\institute{W.O. Krawec \at University of Connecticut\\Department of Computer Science and Engineering\\Storrs CT 06269 USA\\\email{walter.krawec@uconn.edu}}

\maketitle

\begin{abstract}
In this paper, we show an interesting connection between a quantum sampling technique and quantum uncertainty.  Namely, we use the quantum sampling technique, introduced by Bouman and Fehr, to derive a novel entropic uncertainty relation based on smooth min entropy, the binary Shannon entropy of an observed outcome, and the probability of failure of a classical sampling strategy.  We then show two applications of our new relation.  First, we use it to develop a simple proof of a version of the Maassen and Uffink uncertainty relation.  Second, we show how it may be applied to quantum random number generation.
\end{abstract}

\section{Introduction}

In this paper, we revisit a famous entropic uncertainty relation proven by Maassen and Uffink in \cite{MU-bound} (which followed a conjecture by Kraus in \cite{kraus-uncertainty} and was also an improvement of an entropic uncertainty relation first proposed by Deutsch \cite{deutsch-first-bound}).  Given a quantum system $\rho$ and two projective measurements (PMs) $\{M_x\}$ and $\{N_x\}$ (where $M_x = \ket{\mu_x}\bra{\mu_x}$ and $N_y = \ket{\nu_y}\bra{\nu_y}$ for some orthonormal bases $\{\ket{\mu_x}\}$ and $\{\ket{\nu_y}\}$), then one cannot necessarily be certain of the outcome of both measurements.  More specifically, the relation states:
\begin{equation}\label{eq:uncertain}
H(M) + H(N) \ge -\log_2 c,
\end{equation}
where $c$ is a function of the two measurements, namely:
\begin{equation}\label{eq:uncertain-c}
c = \max_{x,y}|\braket{\mu_x|\nu_y}|^2.
\end{equation}
This relation, and numerous others like it (\cite{ent2,ent3,ent4,ent1} just to list a very few), are not only interesting in and of themselves, but also have numerous other applications throughout quantum information science and quantum cryptography.  For a general survey of entropic uncertainty relations, the reader is referred to \cite{survey,survey-2,survey-3}.

In this paper, using a quantum sampling technique introduced in \cite{sampling}, we derive a novel entropic uncertainty relation based on smooth \emph{quantum} min entropy with a direct connection to \emph{classical} sampling strategies.  We use this to derive a novel, and in our opinion simpler, proof of Equation \ref{eq:uncertain} for projective measurements over two-dimensional systems.  We also show how our new bound can be applied to cryptographic applications.  To our knowledge, this sampling technique has not seen application to more broad areas of quantum information before our paper. 

Our new entropic uncertain bound utilizes smooth min entropy and has a direct connection to sampling strategies.  It is also applicable for states which are not necessarily i.i.d.; that is, our result is applicable to arbitrary states and we do \emph{not} need to assume the given state is i.i.d.  This is very useful for cryptographic applications as non i.i.d. states arise when an adversary has the ability to perform an arbitrary general attack on a quantum state; thus the ability for our bound to handle such arbitrary systems means it can be used to prove security for some protocols against general coherent attacks.  This new relation, informally, states that, except with small probability (determined by the user and the dimension of the system), measuring a portion of a system in one basis resulting in outcome $q$ implies the smooth min-entropy in the remaining portion, after measuring in a second basis, can be lower-bounded by the binary Shannon entropy of the Hamming weight of $q$ and the maximal overlap of the two basis measurements, up to some error induced by the sampling technique.  This new relation, which to our knowledge has not been discovered before, may hold interesting applications in quantum cryptography as we discuss later. Furthermore, the techniques we used to derive and prove this new relation may be useful in further extending the quantum sampling technique to other application domains.

There are several contributions in this work.  First, we discover a novel entropic uncertainty bound (involving smooth min entropy and applicable to arbitrary, non-i.i.d. states) directly related to sampling strategies and which may have interesting applications to quantum cryptography and information theory.  We show a rather interesting connection between quantum sampling and quantum uncertainty and use this to derive a much simpler proof of a particular case of Equation \ref{eq:uncertain}.  We also discuss how our methods can be used to analyze certain cryptographic protocols, in particular, quantum random number generators.  Finally, the techniques we use in this paper may find application to other areas of quantum information science and may eventually lead to better bounds for quantum cryptography in the finite key setting.

\subsection{Notation and Definitions}
Let $\mathcal{A}$ be a finite alphabet of size $d$.  Then if $q \in \mathcal{A}^n$ and $\tau = \{\tau_1,\cdots, \tau_k\} \subset \{1,\cdots,n\}$, we write $q_\tau$ to mean the sub-string of $q$ indexed by $\tau$, namely $q_\tau = (q_{\tau_1}, \cdots, q_{\tau_k})$.  We write $q_{-\tau}$ to mean the sub-string of $q$ indexed by the complement of $\tau$.

If $\mathcal{A} = \{0,1\}$, the \emph{Hamming weight} of the string $q$ is defined to be the number of non-zero elements in $q$.  For arbitrary $\mathcal{A}$ and for any $a \in \mathcal{A}$, we define the \emph{relative $a$-Hamming weight}, which we denote by $\hw_a(q)$, to be the number of letters in $q$ \emph{not equal to $a$} and that quantity divided by the length of $q$.  Namely: $\hw_a(q) = {|\{i \st q_i \ne a\}|}/{|q|}$,
where $|q|$ denotes the length of the string $q$.

A \emph{density operator} acting on Hilbert space $\mathcal{H}$ is a Hermitian positive semi-definite operator of unit trace.  Given $\ket{\psi} \in \mathcal{H}$ we write $\kb{\psi}$ to mean $\ket{\psi}\bra{\psi}$.  We define a \emph{Projective Measurement} or PM over a $d$-dimensional Hilbert space $\mathcal{H}$ to be a set of projectors $N = \{\kb{\phi_1}, \cdots, \kb{\phi_{d}}\}$, where $\{\ket{\phi_i}\}_{i=1}^d$ form an orthonormal basis of $\mathcal{H}$.  It is not difficult to see that we may treat a measurement outcome of $\ket{\phi_{j_1}}\otimes\cdots\otimes\ket{\phi_{j_n}}$ as the classical string $j = j_1\cdots j_n$.  We often write $\mathcal{H}_d$ to mean a $d$-dimensional Hilbert space.

We denote $H(X)$ to be the Shannon entropy of random variable $X$.  If $\rho$ is a density operator acting on Hilbert space $\mathcal{H}$ and if $N$ is a PM over $\mathcal{H}$, we write $H(N)_\rho$ to mean the Shannon entropy of the random variable induced by measuring $\rho$ using PM $N$.  Similarly, if $\ket{\psi}$ is a pure state in $\mathcal{H}$ we write $H(N)_\psi$ to mean the entropy of the result of measuring $\kb{\psi}$ using PM $N$.  For technical reasons later, we define an \emph{extended binary entropy function}, denoted $\Hextd(x)$ which is defined to be $H(x, 1-x)$ if $x \in [0,1/2]$; otherwise, if $x < 0$, $\Hextd(x) = 0$ and if $x > 1/2$, then $\Hextd(x) = 1$.

Given a density operator $\rho_{AE}$, acting on some Hilbert space $\mathcal{H}_A\otimes\mathcal{H}_B$, the conditional quantum min entropy \cite{renner2005security}, denoted $\Hmin(A|E)_\rho$, is defined to be:
\begin{equation}\label{eq:min-ent}
\Hmin(A|E)_\rho = \sup_{\sigma_E}\max\{\lambda\in\mathbb{R} \st 2^{-\lambda}I_A\otimes\sigma_E - \rho_{AE} \ge 0\}.
\end{equation}
Here, $I_A$ is the identity operator on $\mathcal{H}_A$ and the notation $X \ge 0$, for some operator $X$, implies that $X$ is positive semi-definite.

To attempt to gain some insight into what, exactly, the above definition means, first consider the case where the $E$ system is trivial.  In this case we may write $\Hmin(A)_\rho$ and it holds that:
\[
\Hmin(A)_\rho = -\log\lambda_{\max}(\rho),
\]
where $\lambda_{\max}(\rho)$ is the maximal eigenvalue of $\rho$ (note that all logarithms in this paper are base $2$ unless otherwise stated).  For classical states, this has a very clear meaning.  Let $\rho_A = \sum_ip_i\kb{i}$ for some orthonormal basis $\{\ket{i}\}$.  Then $\Hmin(A)_\rho$ is simply $-\log\max_ip_i$.  A comparison to von Neumann entropy for the two dimensional case is shown in Figure \ref{fig:comp}.

\begin{figure}
  \centering
  \includegraphics[width=250pt]{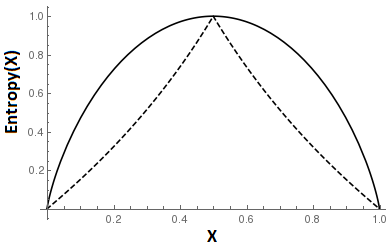}
  \caption{Comparing Shannon entropy (solid) with min entropy of a classical state (dashed) in the two-dimensional case.}\label{fig:comp}
\end{figure}

The more general, conditional min-entropy is more difficult to understand conceptually using only Equation \ref{eq:min-ent}.  Instead, it is more intuitive to think of min entropy in terms of guessing probabilities (\emph{at least, for classical-quantum (cq) states}).  If we have a cq-state of the form $\rho_{AE} = \sum_ip_i\kb{i}\otimes\rho_E^{(i)}$, then it was shown in \cite{min-ent-op} that:
\[
\Hmin(A|E)_\rho = -\log P_{\text{guess}}(\rho_{AE}),
\]
where:
\[
P_{\text{guess}}(\rho_{AE}) = \max_{\{\mathcal{M}_i\}}\sum_i p_itr(\mathcal{M}_i\rho_{E}^{(i)}),
\]
and the maximum is over all POVM operators on $\mathcal{H}_E$.
Thus, for cq-states at least, one can think of min-entropy in terms of ``guessing games.''  This will not be important to our discussion, however it helps to give a clearer picture of what, exactly, min-entropy is measuring.

Quantum min-entropy has many applications in quantum cryptography, especially in finite-key scenarios.  In particular, given a cq-state $\rho_{AE}$ (perhaps derived from some quantum cryptographic protocol), where the $A$ register is correlated with the $E$ register in some way.  One may apply \emph{privacy amplification} to attempt to establish a uniform random string independent of $E$'s quantum register.  Let $\sigma_{KE'}$ be the resulting cq-state after processing $\rho_{AE}$ through privacy amplification (essentially, \emph{publicly} choosing a random two-universal hash function, and applying it to the $A$ register).  The $K$ register is of size $\ell$ bits and the $E'$ register contains $E$'s original information plus the hash function used.  In \cite{renner2005security}, it was shown that:
\begin{equation}\label{eq:PA}
\tdl \sigma_{KE'} - I_K/2^\ell\otimes\sigma_{E'}\tdr \le 2^{-\frac{1}{2}(\Hmin(A|E)_\rho - \ell)}.
\end{equation}
Thus, deriving bounds on min-entropy is highly useful as they lead directly to bounds on how large a random string may be distilled from a given cq-state (they also may be used for quantum key distribution, though there one must also take into account the information leaked during error correction).  We will return to this in a later section.

For notation, if $N$ is a PM on $\mathcal{H}$ and $\rho$ is a density operator on $\mathcal{H}^{\otimes n}$, then we use $\Hmin(N)_\rho$ to mean the min entropy of the resulting state following the measurement of each of the $n$ sub-spaces $\rho$ acts on using PM $N$.  If $p(j)$ is the probability of observing outcome $j = j_1\cdots j_n$ (i.e., after measuring, one observes the quantum state $\ket{\phi_{j_1}}\otimes\cdots\otimes\ket{\phi_{j_n}}$) it is not difficult to see that:
$\Hmin(N)_\rho = -\log\max_jp(j).$

Given a density operator $\rho_{AC}$ acting on $\mathcal{H}_A\otimes\mathcal{H}_C$, where the $C$ portion is \emph{classical} (namely, we may write $\rho_{AC} = \sum_cp_c\sigma_{A}^{(c)}\otimes\kb{c}$, where $\{\ket{c}\}$ is an orthonormal basis of $\mathcal{H}_C$ and each $\sigma_{AB}^{(c)}$ is an arbitrary density operator acting on $\mathcal{H}_A$) then the conditional min entropy $\Hmin(A|C)_\rho$ is:
\begin{equation}\label{eq:min-ent-classical}
\Hmin(A|C)_\rho \ge \inf_{c}\Hmin(A)_{\sigma^{(c)}},
\end{equation}
The above can be proven from Lemma 3.1.8 in \cite{renner2005security} and the definition of conditional min entropy.

Finally, the \emph{$\epsilon$-smooth min entropy}, denoted $\Hmin^\epsilon(\rho)$ is defined to be:
\begin{equation}\label{eq:smooth-entropy}
\Hmin^\epsilon(\rho) = \sup_{\sigma \in \Gamma_\epsilon(\rho)}\Hmin(\sigma),
\end{equation}
where $\Gamma_\epsilon(\rho)$ is the set of all density operators $\epsilon$ close to $\rho$ as measured by the trace distance; i.e.,
\begin{equation}\label{eq:Gamma}
\Gamma_\epsilon(\rho) = \{\sigma \st \trd{\sigma - \rho} \le \epsilon\},
\end{equation}
and $\trd{A}$ is the \emph{trace distance} of $A$.  We define $\Hmin^\epsilon(N)_\rho$ similarly to $\Hmin(N)_\rho$ described earlier whenever $N$ is a PM.  The conditional smooth entropy, $\Hmin^\epsilon(A|B)$ is defined similarly.  Note that there is a version of privacy amplification (Equation \ref{eq:PA}) for smooth min entropy, proven in \cite{renner2005security}, which we will use later:
\begin{equation}\label{eq:PA-smooth}
\tdl \sigma_{KE'} - I_K/2^\ell\otimes\sigma_{E'}\tdr \le 2^{-\frac{1}{2}(\Hmin^\epsilon(A|E)_\rho - \ell)} + 2\epsilon.
\end{equation}

An important result, which we will use later, was proven in \cite{sampling} (based on a Lemma in \cite{renner2005security}) and allows one to compute the min entropy of a superposition of states:
\begin{lemma}\label{lemma:super-entropy}
(From \cite{sampling}): Let $\mathcal{H}$ be a $d$-dimensional Hilbert space with orthonormal basis $\{\ket{i}\}_{i=1}^d$ and let $\mathcal{H}_E$ be an arbitrary finite dimensional Hilbert space.  Then, for any pure state $\ket{\psi} = \sum_{i\in J}\alpha_i\ket{i}\otimes\ket{\phi_i}_E \in \mathcal{H}\otimes\mathcal{H}_E$, if we define:
\begin{align*}
\rho = \sum_{i\in J}|\alpha_i|^2\kb{i}\otimes\kb{\phi_i}_E,
\end{align*}
it holds that for any PM $N$ on $\mathcal{H}$:
\begin{equation}
\Hmin(N|E)_\psi \ge \Hmin(N|E)_\rho - \log_2|J|.
\end{equation}
\end{lemma}
The above lemma will allow us to bound the min entropy of a \emph{superposition} of states, by computing, instead, the min entropy in a suitable \emph{mixed} state.

\section{Quantum Sampling}

Since our proof relies on the quantum sampling technique introduced in \cite{sampling}, we now review this subject here.  All information in this section is derived from \cite{sampling} (we make only a few changes in notation and some generality) and is meant only as a review of this material for completeness.

Let $\mathcal{A}$ be a finite alphabet of size $d$, and let $a \in \mathcal{A}$, and $k \in \mathbb{N}$.  We assume $d$, $a$, and $k$ are arbitrary, but fixed.  A sampling strategy is a pair $\Sigma = (P^k_T, \mathcal{F}_{a}^k)$ where $P_T^k$ is a distribution over all subsets of $\{1,\cdots,n\}$ of size $k$ and $\mathcal{F}_a^k$ is a function which, given a subset of a sample $q \in \mathcal{A}^n$ (i.e., given $q_\tau$), will output a guess of the value $\hw_a(q_{-\tau})$.  That is, given a randomly chosen sample $q_\tau$ (where $\tau$ was drawn according to $P_T^k$), $\mathcal{F}_a^k$ will estimate the value of $\hw_a$ in \emph{the remaining portion} of $q$.  When it is clear, we will often forgo writing the superscript, and simply write $\mathcal{F}_a$.

Define $B_{\tau,a}^\delta(\Sigma)$ to be the set of all words in $\mathcal{A}^n$ such that the estimate provided by $\mathcal{F}_a$ is $\delta$ close to the actual value given a fixed subset $\tau \subset\{1,\cdots,n\}$ of size $k$.  That is, let:
\[
B_{\tau,a}^\delta(\Sigma) = \{q\in\mathcal{A}^n \st |\mathcal{F}_{a}(q_\tau) - \hw_a(q_{-\tau})| \le \delta\}.
\]
Informally, if we have a fixed subset $\tau$ with $|\tau| = k$, then the set $B_{\tau,a}^\delta(\Sigma)$ defines the set of all ``good'' strings; i.e., strings for which the sampling strategy $\Sigma$ provides an accurate estimate of $\hw_a$, up to an error of $\delta$ assuming $\tau$ was the chosen subset.

From this, the \emph{error probability of $\Sigma$} is defined to be:
\begin{equation}
\epsilon_\delta^{cl} = \max_{q\in\mathcal{A}^n}Pr(q \not\in B_{T, a}^\delta(\Sigma)).
\end{equation}
where the probability is over all subsets $\tau$ chosen according to $P_T^k$ (i.e., we treat $B_{T,a}^\delta$ as a random variable induced by choosing subsets $\tau$ according to $P_T^k$).  From this definition, it is clear that for any word $q \in \mathcal{A}^n$, the estimated value of $\hw_a$, given by the sampling strategy $\Sigma$, is $\delta$ close to the real value in the remainder of the string (i.e., in the portion of the string that was not used in the test set $\tau$), except with probability $\epsilon_\delta^{cl}$.  Note the superscript ``$cl$'' is used to show this is the error probability of a \emph{classical} sampling strategy.

One important sampling strategy we will make use of is the following: Let $P_T^k$ be the uniform distribution over all subsets $\tau \subset \{1,\cdots,N\}$ with $|\tau| = k$; i.e., $Pr(P_T^k = \tau) = 1/{N\choose k}$.  Then, given a string $q \in \mathcal{A}^N$, the function $\mathcal{F}$ is defined simply to be: $\mathcal{F}_a(q_\tau) = \hw_a(q_\tau)$.  That is, the sampling strategy is to choose a random subset, uniformly at random, evaluate $\hw_a$ on that subset, and output, as an estimate of the value $\hw_a(q_{-\tau})$, the value $\hw_a(q_\tau)$.  The following Lemma was proven in \cite{sampling} (see Appendix B in the extended, online version, of that reference):
\begin{lemma}\label{lemma:sample}
(From \cite{sampling}): Let $\delta > 0$ be given and $\Sigma$ be as described above in the text.  If $|\tau|=k\le N/2$ then for any $d$ and $a$, it holds that:
\begin{equation}\label{eq:lemma-sample-error}
\epsilon_\delta^{cl} \le 2\exp\left(-\frac{\delta^2kN}{N+2}\right).
\end{equation}
\end{lemma}

These notions can be extended to the quantum domain \cite{sampling}.  Consider an orthonormal basis $\{\ket{a} \st a \in \mathcal{A}\}$ and let $\mathcal{H}_A$ be the $d$-dimensional Hilbert space spanned by this basis.  Let $U$ be a unitary operator acting on $\mathcal{H}_A$.  Then, we may define an orthonormal basis:
\[
\mathcal{B} = \{U^{\otimes n}\ket{b_1\cdots b_n} = U\ket{b_1}\otimes\cdots\otimes U\ket{b_n} \st b_i\in\mathcal{A}\},
\]
of the Hilbert space $\mathcal{H}_A^{\otimes n}$.  Then, given a state $\ket{\psi} \in \mathcal{H}_A^{\otimes n}\otimes \mathcal{H}_E$, it is said to have relative $a$-Hamming weight $\beta$ \emph{in $A$} \emph{with respect to basis $\mathcal{B}$}, if we can write $\ket{\psi} = U^{\otimes n}\ket{b_1\cdots b_n}\otimes\ket{\phi}_E$ with $\hw_a(b) = \beta$.  Note that we are allowed an additional, arbitrary, system in some Hilbert space $\mathcal{H}_E$ (this may be the trivial space if it is not needed).  Also, notice that this definition is dependent on the choice of basis.

By abusing notation slightly, we may also define span$(B_{\tau,a}^\delta)$ to be:
\[
\text{span}\left(\{U^{\otimes n}\ket{q} \st q \in \mathcal{A}^n \text{ and } |w_a(q_\tau) - w_a(q_{-\tau})| \le \delta\}\right)
\]
Note that if $\ket{\psi} \in \text{span}(B_{t,a}^\delta)\otimes \mathcal{H}_E$ then, if sampling is done by measuring in the $\mathcal{B}$ basis on subset $\tau$, it is guaranteed that the state collapses to a superposition of states which are $\delta$ close to the observed $a$-Hamming weight (with respect to basis $\mathcal{B}$).  Also note we will drop the $\delta$ superscript when the context is clear.

Using the above definitions, the main result from \cite{sampling} is as follows:
\begin{theorem}\label{thm:sample}
(From \cite{sampling}, though reworded for our application in this paper and our specific sampling strategy): Let $k \le n/2$ be given and consider sampling strategy $\Sigma$ as described above.  Then, for every pure state $\ket{\psi} \in \mathcal{H}_A^{\otimes n}\otimes\mathcal{H}_E$, there exists a collection of ``ideal states'' $\{\ket{\phi^\tau}\}$ where the index is over all subsets $\tau$ of size $k$ and each $\ket{\phi^\tau} \in \text{span }(B_{\tau,a}^\delta)\otimes\mathcal{H}_E$ such that:
\[
\left|\left| \frac{1}{T}\sum_{\tau}\kb{\tau}\otimes\kb{\psi} - \frac{1}{T}\sum_\tau\kb{\tau}\otimes\kb{\phi^\tau}\right|\right| \le \sqrt{\epsilon_\delta^{cl}},
\]
where $T = {n \choose k}$ and the sum is over all subsets of size $k$. Note that we prepend an auxiliary system spanned by orthonormal basis $\{\ket{\tau}\}$ for all appropriate subsets $\tau$.
\end{theorem}

The above result states that, on average over the choice of subset $\tau$, the real system $\ket{\psi}$ is $\epsilon$-close to an ideal state, where the ideal state is defined to be one where the sampling strategy always works (i.e., where, after sampling, regardless of the subset choice, the state collapses to one which is a superposition of states $\delta$ close to the estimate).  Furthermore, $\epsilon$ can be computed from the classical error probability.  

\section{Main Result}

We are now in a position to state, and prove, our new entropic uncertainty relation.

\begin{theorem}\label{thm:main-qubit}
Let $\hat\epsilon\ge\epsilon > 0$, $a\in\{0,1\}$, $0<\beta<1/2$, and $\rho$ a density operator acting on Hilbert space $\mathcal{H}_2^{\otimes(m+n)}$ with $m \le n$ be given.  Also, let $\PM = \{\kb{\mu_0}, \kb{\mu_1}\}$ and $\PN = \{\kb{\nu_0}, \kb{\nu_1}\}$, be two projective measurements.  If a subset $t$ of size $m$ of $\rho$ is measured using $\PM$ resulting in outcome $q$ we denote by $\rho(t,q)$ to be the post measurement state (this is well defined given $\rho$).
Then it holds that:
\[
Pr\left[\Hmin^{2\epsilon+2\epsilon^\beta}(N)_{\rho(t,q)} + n\Hextd(w_a(q+\delta)) \ge -n\log c)\right] \ge 1-\hat\epsilon^{1-2\beta}
\]
where the probability is over all choice of subsets and resulting measurement outcomes.
Above, $c$ is defined in Equation \ref{eq:uncertain-c} and:
\begin{equation}\label{eq:thm1-delta}
\delta = \sqrt{\frac{(m+n+2)\ln(2/\epsilon^2)}{m(m+n)}}.
\end{equation}
\end{theorem}
\begin{proof}
We first consider the case when $\rho$ is pure; that is, $\rho = \kb{\psi}$ for some $\ket{\psi} \in \mathcal{H}_2^{\otimes(m+n)}$.  Then, applying Theorem \ref{thm:sample} to $\rho$, using the sampling strategy described in the previous section for a sample subset size of $m$, it follows that there exists an ``ideal'' state $\sigma$ of the form:
$\sigma = \frac{1}{T}\sum_{t} \kb{t}\otimes \kb{\phi^t},$
where $T$ is the number of possible subsets (i.e., $T = {n+m \choose m}$); the summation is over all possible subsets $t$ of $\{1,\cdots,n+m\}$ which are of size $m$ (we expand the underlying Hilbert space to include this auxiliary subspace $\mathcal{H}_T$ spanned by orthonormal basis $\{\ket{t}\st t \subset\{1,\cdots,n+m\}, |t|=m\}$; and, finally, each $\ket{\phi^t} \in \text{span } (B_{t,a}^\delta)$.  This ideal state satisfies the following:
\[
\left|\left|\sigma - \frac{1}{T}\sum_t\kb{t}\otimes \kb{\psi}\right|\right| \le \sqrt{\epsilon_\delta^{cl}}.
\]
Given $\delta$ as in Equation \ref{eq:thm1-delta}, and also given Lemma \ref{lemma:sample}, it holds that $\sqrt{\epsilon_\delta^{cl}} = \epsilon$.

Consider the following experiment: First, run the sampling strategy, choosing a random subset $t$ (which is chosen by measuring the auxiliary $\mathcal{H}_T$ subspace) and performing a measurement in the $\PM$ basis resulting in outcome $q$ (note that $q$ depends on the subset chosen and the intrinsic randomness of the measurement itself).  Let $\rho(t,q)$ be the post-measurement state if this experiment is performed on the true state $\rho = \kb{\psi}$.  Likewise, let $\sigma(t,q)$ be the post measurement state if this experiment is performed on the ideal state $\sigma$.  Both post-measurement states are well defined given both $t$ and $q$ (though, of course, the post-measurement state may be a superposition, they are, however, exactly defined pure states, conditioning on the outcome of $t$ and $q$).

We first show:
\begin{equation}\label{eq:thm1:rev1}
\Hmin(N)_{\sigma(t,q)} \ge -n\log c - n\Hextd(w_a(q)+\delta).
\end{equation}
That is, with certainty, for any subset $t$ and observed value $q$, Equation \ref{eq:thm1:rev1} holds in the \emph{ideal case}.

Let $t$ be the chosen subset, thus the measurement in basis $\PM$ is performed on the pure state $\ket{\phi^t}$.  Since $\ket{\phi^t} \in \text{span }(B_{t,a}^\delta)$, it follows that the post measurement state, after observing value $q$, collapses to a superposition of the form:
\begin{equation}\label{eq:thm1-super}
\ket{\phi'} = \sum_{i\in J} \alpha_i\ket{\mu_{i_1},\cdots,\mu_{i_n}},
\end{equation}
where $J \subset I = \{i \in \{0,1\}^n \st |w_a(i) - w_a(q)| \le \delta\}$ and normalization requires $\sum_i|\alpha_i|^2 = 1$.  Of course $\sigma(t,q) = \kb{\phi'}$.

Now, consider the mixed state:
\[
\chi = \sum_{i\in J}|\alpha_i|^2\kb{\mu_{i_1}\cdots, \mu_{i_n}}.
\]
By applying Lemma \ref{lemma:super-entropy}, we have:
\begin{equation}\label{eq:thm1-lemma-bound}
\Hmin(N)_{\sigma(t,q)} = \Hmin(N)_{\phi'} \ge \Hmin(N)_\chi - \log|J|.
\end{equation}
We now compute $\Hmin(N)_\chi$.  Let $\chi_N$ be the result of measuring $\chi$ using PM $\PN$.  It is not difficult to see that this state is simply:
\begin{align*}
\chi_N &= \sum_{i\in J}|\alpha_i|^2\left(\sum_{j\in\{0,1\}^n}p(j|i)\kb{\nu_{j_1},\cdots,\nu_{j_n}}\right)\\
&= \sum_{j\in \{0,1\}^n}p(j)\kb{\nu_{j_1},\cdots,\nu_{j_n}},
\end{align*}
where we define $p(j|i) = p(j_1\cdots j_n|i_1\cdots i_n)$ to be the probability of observing $\ket{\nu_{j_1}\cdots \nu_{j_n}}$ if given an input state of $\ket{\mu_{i_1}\cdots\mu_{i_n}}$.  We define $p(j) = \sum_{i \in J}|\alpha_i|^2p(j|i)$.  It is straight-forward to compute $p(j|i)$:
\begin{equation}\label{eq:thm1-pji}
p(j|i) = p(j_1\cdots j_n|i_1\cdots i_n) = \prod_{l=1}^n|\braket{\nu_{j_l}|\mu_{i_l}}|^2
\end{equation}
Since $\chi_N$ is a classical system, we have:
\[
\Hmin(N)_\chi = -\log\max_jp(j) = -\log\max_j\left[\sum_{i\in J}|\alpha_i|^2p(j|i)\right].
\]
Let $p^* = \max_{i,j}p(j|i)$ (where the maximum is over all $i \in J$ and $j \in \{0,1\}^n$).  Then it is clear that:
\[
\max_jp(j) = \max_j\left[\sum_{i\in J}|\alpha_i|^2p(j|i)\right] \le p^*,
\]
(recall that $\sum_i|\alpha_i|^2 = 1$) and thus:
\[
\Hmin(N)_\chi = -\log\max_jp(j) \ge -\log p^*.
\]
Finally, we compute a bound on $p^*$ as:
\[
p^* = \max_{\substack{j\in\{0,1\}^n\\i\in J}} \prod_{l=1}^n |\braket{\nu_{j_l}|\mu_{i_l}}|^2 \le c^n,
\]
where $c = \max_{x,y}|\braket{\nu_x|\mu_y}|^2$.  Thus:
\begin{equation}\label{eq:thm1-ent-bound}
\Hmin(N)_\chi \ge -\log p^* \ge -n\log c.
\end{equation}
It is clear that $J \subset \{i\in\{0,1\}^n \st w_a(i) \le w_a(q) + \delta\}$ and so using the well-known bound on the volume of a Hamming ball we have $|J| \le 2^{n\Hextd(w_a(q) + \delta)}$ (note we are using our ``extended'' version $\Hextd$ here to avoid the issue when $\hw_a(q)+\delta > 1/2$; indeed, if that is the case then $\Hextd(\cdot) = 1$ and so the bound holds trivially), we may combine this with Equations \ref{eq:thm1-lemma-bound} and \ref{eq:thm1-ent-bound} to derive:
\[
\Hmin(N)_{\sigma(t,q)} \ge -n\log c - n\Hextd(w_a(q)+\delta).
\]

Of course, the above analysis only considered the ideal state from which we are guaranteed that the sampling strategy was successful.  We now consider the ``real'' state $\rho = \kb{\psi}$.

Consider the real state $\frac{1}{T}\sum_t\kb{t}\otimes\kb{\psi}$.  The process of choosing a subset $t$, measuring, and observing $q$ (resulting in post-measurement state $\rho(t,q)$) may be described, entirely, by the mixed state:
$\rho_{TQR} = \frac{1}{T}\sum_t \kb{t}\sum_q p(q|t)\kb{q}\otimes\rho(t,q),$
where $p(q|t)$ is the probability of observing outcome $q$ given subset $t$ was sampled; here we use ``$R$'' to denote the ``remainder'' - that is the portion of the state not yet measured.  Likewise, the ideal state, after performing this experiment, may be written as the mixed state:
$\sigma_{TQR} = \frac{1}{T}\sum_t \kb{t}\sum_q \tilde{p}(q|t)\kb{q}\otimes\sigma(t,q).$
Since quantum operations cannot increase trace distance, we have $||\rho_{TQR}-\sigma_{TQR}||\le\epsilon$. By basic properties of trace distance:
\begin{equation}\label{eq:tr-dist-new}
\epsilon \ge \frac{1}{T}\sum_t\sum_q|| p(q|t)\rho(t,q) - \tilde{p}(q|t)\sigma(t,q)||.
\end{equation}
Of course, it holds that $\frac{1}{T}\sum_t\sum_q | p(q|t) - \tilde{p}(q|t)| \le \epsilon$ (this follows by tracing out the unmeasured portion ``$R$'' of $\rho_{TQR}$ and $\sigma_{TQR}$ and again realizing that quantum operations, such as partial trace, do not increase trace distance).  Let $\tilde{p}(q|t) = p(q|t) + \epsilon_{q,t}$ where $\epsilon_{q,t}$ may be positive or negative.  Then, the above inequality of course implies $\frac{1}{T}\sum_t\sum_q|\epsilon_{q,t}| \le \epsilon$.

Returning to Equation \ref{eq:tr-dist-new} we then find:
\begin{align}
\epsilon &\ge \frac{1}{T}\sum_t\sum_q||p(q|t)(\rho(t,q) - \sigma(t,q)) - \epsilon_{q,t}\sigma(t,q)||\notag\\
&\ge \sum_t\sum_q p(q\wedge t)2\cdot\Delta_{q,t} - \epsilon,
\end{align}
where we define $\Delta_{q,t} = \frac{1}{2}||\rho(t,q) - \sigma(t,q)|| \le 1$.  Note that, above, we made use of the reverse triangle inequality and the fact that $||\sigma(t,q)|| = tr\sigma(t,q) = 1$ since $\sigma(t,q)$ is a positive operator of unit trace.  We also used the fact that $p(q\wedge t) = p(q|t)p(t) = p(q|t)\cdot\frac{1}{T}$ (here, $p(q\wedge t)$ is the probability of sampling subset $t$ and observing $q$).  Of course, the above implies:
\begin{equation}
\sum_{t,q}p(q\wedge t)\Delta_{q,t} \le \epsilon.
\end{equation}

Now, let us consider $\Delta_{q,t}$ as a random variable over the choice of all subsets $t$ and measurement outcomes on that subset $q$.  The expected value is easily seen to be $\mathbb{E}(\Delta_{q,t}) = \mu \le \epsilon$.  We also compute the variance $V^2$:
\begin{align*}
V^2 &= \sum_{q,t}p(q \wedge t)\Delta_{q,t}^2 - \mu^2 \le \sum_{q,t}p(q \wedge t)\Delta_{q,t} - \mu^2\notag\\
&= \mu(1-\mu) \le \mu \le \epsilon,
\end{align*}
where, above, we used the fact that $\Delta_{q,t} \le 1$ and so $\Delta_{q,t}^2 \le \Delta_{q,t}$.

Now, by Chebyshev's inequality, we have:
\begin{align}
Pr\left( |\Delta_{q,t} - \mu| \ge \epsilon^{\beta}\right) \le \frac{V^2}{\epsilon^{2\beta}} \le \epsilon^{1-2\beta} \le \hat\epsilon^{1-2\beta},
\end{align}
(the last inequality follows since $\beta < 1/2$); note that this probability is over all subsets $t$ and measurement outcomes $q$.  Thus, except with probability at most $\hat\epsilon^{1-2\beta}$, after choosing $t$ and observing $q$, it holds that $|\Delta_{q,t}-\mu| \le \epsilon^\beta$ which, of course, implies:
\[
\frac{1}{2}||\rho(t,q) - \sigma(t,q)|| = \Delta_{t,q} \le \mu+\epsilon^\beta \le \epsilon+\epsilon^\beta.
\]
Since, in this case we have $\sigma(t,q) \in \Gamma_{2\epsilon+2\epsilon^\beta}(\rho(t,q))$, it holds:
\[
\Hmin^{2\epsilon+2\epsilon^\beta}(N)_{\rho(t,q)} \ge \Hmin(N)_{\sigma(t,q)} \ge -n\log c - \Hextd(w_a(q) + \delta),
\]
completing the proof when the case $\rho$ is pure.

Now consider the case when $\rho$ is not pure.  In this case, let $\ket{\psi}_{HC}$ be a purification of $\rho$, where the $H$ portion is the original $\mathcal{H}_2^{\otimes (m+n)}$ space and the $C$ portion lives in an extra Hilbert space ($\mathcal{H}_C$) needed to purify $\rho$.  As before, using quantum sampling, there exists an ideal state $\sigma$ where, now, each of the $\ket{\phi^t} \in \text{span } (B_{t,a}^\delta)\otimes\mathcal{H}_C$.

Let us consider running the same experiment as before on this ideal state (where, now, the experiment consists only of measuring the $H$ portion, not the $C$ portion).  Let $t$ be the chosen subset and $q$ the observed value.  Then, in the ideal case, the state collapses to a pure state of the form:
\[
\ket{\phi'}_{HC} = \sum_{i\in J}\alpha_i\ket{\mu_{i_1},\cdots, \mu_{i_n}}\otimes \ket{C_i},
\]
where $J$ is defined as before and the states $\ket{C_i}$ are arbitrary (not necessarily orthogonal) states in $\mathcal{H}_C$.  Let $\chi_{HC} = \sum_{i\in J}|\alpha_i|^2\kb{\mu_{i_1}, \cdots, \mu_{i_n}}\otimes\kb{C_i}$.  From Lemma \ref{lemma:super-entropy}, we have:
\[
\Hmin(N|C)_{\phi'} \ge \Hmin(N|C)_\chi - \log|J|.
\]
We add an additional system $I$ spanned by orthonormal basis $\{\ket{I_i}\}_{i\in J}$ and define the following state:
\[
\chi_{HCI} = \sum_{i\in J}|\alpha_i|^2\kb{\mu_{i_1},\cdots, \mu_{i_n}}\otimes\kb{C_i}\otimes\kb{I_i}
\]
Measuring this state using PM $\PN$ yields:
\[
\chi_{NCI} = \sum_{i\in J}|\alpha_i|^2\kb{I_i}\otimes\kb{C_i}\otimes\sum_{j\in\{0,1\}^n}p(j|i)\kb{\nu_{j_1},\cdots, \nu_{j_n}},
\]
where $p(j|i)$ is defined as before in Equation \ref{eq:thm1-pji} (also, note that we permuted the ordering of the sub-spaces above only for clarity).  Define the states $\chi_{N,i}$ as:
\[
\chi_{N,i} = \sum_{j\in\{0,1\}^n}p(j|i)\kb{\nu_{j_1},\cdots,\nu_{j_n}}.
\]
from which we may write $\chi_{NCI} = \sum_{i\in J}|\alpha_i|^2\kb{I_i,C_i}\otimes\chi_{N,i}$.

Thinking of the $CI$ system jointly, the above state is classical on this joint $CI$ system; thus, from Equation \ref{eq:min-ent-classical}, we have:
\begin{align*}
\Hmin(N|CI)_\chi &\ge \inf_{i \in J} \Hmin(N)_{\chi_{N,i}}\\
&= \inf_i(-\log\max_jp(j|i))\\
&\ge -\log p^*\ge -n\log c.
\end{align*}
Finally, from the strong subadditivity of min entropy \cite{renner2005security}:
\begin{align*}
\Hmin(N)_{\phi'} &\ge H(N|C)_{\phi'} \ge \Hmin(N|C)_\chi - \log|J|\\
&\ge \Hmin(N|CI)_\chi - \log|J|\\
&\ge -n\log c - \log|J|\\
&\ge -n\log c - n\Hextd(w_a(q)+\delta),
\end{align*}


The above analysis only utilized the ideal state from which sampling is guaranteed to succeed.  However, the analysis of the real state follows identically as earlier (when we considered an initial pure state), thus completing the proof.
\end{proof}

\section{Applications}
Our Theorem \ref{thm:main-qubit} gives us an interesting entropic uncertainty bound in terms of smooth entropy and also in terms of the success of a classical sampling strategy.  Beyond its independent interest, we show two applications of our new entropic uncertainty result.  First, it gives us a new proof of the Maassen and Uffink entropy relation.  Second, we can apply it to the analysis of source-independent quantum random number generation protocols against adversarial, but memoryless, sources.

\subsection{Application One: Maassen and Uffink Entropic Uncertainty}
As a simple corollary, our Theorem \ref{thm:main-qubit} gives us the usual Maassen and Uffink entropic relation.
\begin{corollary}
Let $\PM$ and $\PN$ be two PMs and $\rho$ a qubit density operator.  Then, except with arbitrarily small probability, it holds that:
\[
H(M)_\rho + H(N)_\rho \ge -\log c.
\]
\end{corollary}
\begin{proof}
Let $\rho$ be a density operator on $\mathcal{H}_2$ and consider the state $\rho' = \rho^{\otimes 2n}$.  Let $a = \max_xtr(\kb{\mu_x}\cdot\rho)$; in particular, if measuring $\rho$ using $\PM$ the probability of observing $\ket{\mu_a}$ is no less than $1/2$.  Note that this ``$a$'' need not be known to users making the measurement, however it clearly exists.  Since $\rho'$ is i.i.d., for any subset $t$ of size $n$ and any measurement outcome $q$ on that subset, the post-measurement state is simply $\rho^{\otimes n}$.

Fix $\hat\epsilon > 0$ and $0<\beta<1/2$.  Then, \emph{for any} $n$ and $\epsilon \le \hat\epsilon$, Theorem \ref{thm:main-qubit}  implies that, except with probability at most $\hat\epsilon^{1-2\beta}$, the following inequality holds:
\begin{equation}\label{eq:cor1-bound1}
\frac{1}{n}\Hmin^{2\epsilon+2\epsilon^\beta}(N)_{\rho^{\otimes n}} + \Hextd(w_a(q)+\delta) \ge -\log c,
\end{equation}
where $q$ is the observed value after measuring using $\PM$ and:
\[
\delta = \sqrt{\frac{(n+1)\ln(2/\epsilon^2)}{n^2}}.
\]
(We used $m=n$ when applying the theorem.)
By the asymptotic equipartition property \cite{aep}, we have $\lim_{\epsilon\rightarrow 0}\lim_{n\rightarrow\infty}\frac{1}{n}\Hmin^{2\epsilon+2\epsilon^\beta}(N)_{\rho^{\otimes n}} = H(N)_\rho$.  By the law of large numbers, we have $\lim_{n\rightarrow \infty}w_a(q) = p_{1-a}$.  Note that by definition of $a$, we have $p_{1-a} \le 1/2$ thus allowing us to replace $\Hextd(\cdot)$ with $H(p_{1-a}, p_a) = H(M)_\rho$. Finally, $\delta \rightarrow 0$ as $n\rightarrow \infty$.  Given fixed $\hat\epsilon$ the above holds; of course $\hat\epsilon$ may be made arbitrarily small, thus yielding the result.
\end{proof}

\subsection{Second Application: Random Number Generation}

We show in this section an interesting application of our new entropic uncertainty relation derived in Theorem \ref{thm:main-qubit} to quantum random number generation in the source independent model.  The goal of a quantum random number generator (QRNG) is to utilize quantum physical properties (e.g., random measurement outcomes) to produce true randomness useful for numerous other tasks (including for cryptography).  Several security models exist ranging from the very weak fully-trusted scenario to the very strong device independent (DI) model \cite{qrng-di1,qrng-di2} (which, though having strong security guarantees, is slow to implement in practice \cite{di-exp,di-exp2}).  In between is the \emph{source independent} (SI) model whereby only the source is untrusted, but the measurement devices are characterized \cite{si-qrng1,qrng-non-iid,qrng-high-speed-cv,qrng-high-speed-cv2}.  See \cite{qrng-survey} for a general survey of QRNGs and their security models.

We show that our new entropic uncertainty relation, proven in Theorem \ref{thm:main-qubit}, has applications to this cryptographic protocol.  This is only preliminary work to show the potential usefulness of quantum sampling applied to broader quantum information science and cryptography and, so, the model we consider is a memory-less adversarial source.  This source, controlled by an adversary, prepares a general $N$ qubit state and sends it to user $A$.  An honest source should prepare the state $\ket{+}^{\otimes N}$ but an adversarial source may prepare anything - we do not require any assumptions on the overall structure of this state beyond that it consists of $N$ qubits and it may even be non -i.i.d.  This user chooses a random sample of size $m$ (this requires some initial private randomness, thus the QRNG must actually extend this initial seed randomness and it's usage must be taken into account) and measures in a test basis (for our sake, we use the $X=\{\ket{+}, \ket{-}\}$ basis) observing outcome $q$ (as a bitstring - if there is no noise and the source is honest, $q = 0^m$).  The remaining $n = N-m$ qubits are measured in the $Z = \{\ket{0}, \ket{1}\}$ basis.  Following this, privacy amplification may be run to distill an $\ell$-bit random string.  Using privacy amplification (see Equation \ref{eq:PA-smooth} but the $E$ system is trivial here as we consider a memory-less adversary), we have:
\begin{equation}
\tdl\rho_{R} - I_R/2^\ell \tdr \le 2\epsilon' + 2^{-\frac{1}{2}(\Hmin^{\epsilon'}(A) - \ell)} = \epsilon_{PA}.
\end{equation}
Above, $\rho_{A}$ is the state of the $n$ measurement results in the $Z$ basis \emph{before} privacy amplification and $\rho_R$ is the state after privacy amplification (transforming the $A$ register of size $n$ to the $R$ register of size $\ell$).
Thus, if we want the trace distance to be no greater than a given $\epsilon_{PA}$ (giving us an $\epsilon_{PA}$-random string), we have:
\[
\ell = \Hmin^{\epsilon'}(A|E)_\rho - 2\log\left(\frac{1}{\epsilon_{PA} - 2\epsilon'}\right).
\]
(Note we require $\epsilon_{PA} > 2\epsilon'$, where $\epsilon'$ is whatever smoothening parameter is used.)
Interestingly, while the choice of the random hash function used for privacy amplification must be random, it was proven in \cite{QRNG-1PA} that once chosen it can be fixed and so we do not need to use additional randomness to choose a hash function (it could be chosen randomly once and then hard-coded into $A$'s device - see \cite{QRNG-1PA} for more details).

If the adversary prepares $N$ qubit states, unentangled with any quantum memory, then we may immediately use our Theorem \ref{thm:main-qubit} to compute $\ell$.  Indeed, let $\epsilon > 0$ and $\beta \in (0,1/2)$ be given.  Let $\epsilon_{PA} = 5\epsilon + 4\epsilon^\beta$.  Then, using the $Z$ and $X$ basis, where $c = 1/2$, we have, except with a failure probability of $\epsilon^{1-2\beta}$, after privacy amplification the size of the final random string is:
\[
\ell_{QRNG} = n(1 - \Hextd(w(q) + \delta))) - \log\frac{1}{\epsilon}.
\]
where $q$ is the observed bit string on the $m$ test qubits (measured in the $X$ basis), and where $\delta$ is given in Equation \ref{eq:thm1-delta}.  Note that the choice of $\beta$ factors into $\epsilon_{PA}$ (which determines how close the output is to uniform randomness) and the failure probability of the entire protocol.  Of course both terms may be made arbitrarily small, but note that, for \emph{fixed} $\epsilon$, as $\beta$ decreases, the failure probability decreases, while $\epsilon_{PA}$ increases.  This choice of $\beta$ is something users may optimize over.

Of course, we must also take into account the randomness used to choose a random subset of size $m$.  This requires $\log{N\choose m}$ bits.  Thus, the total size of the final random string, after sacrificing these initial seed bits, is:
\begin{equation}
\ell_{QRNG} = n(1 - \Hextd(w(q) + \delta))) - \log\frac{1}{\epsilon} - \log{N \choose m}.
\end{equation}
The random bit generation rate is simply $\ell_{QRNG} / N = \ell_{QRNG} / (n+m)$.

We set $\epsilon = 10^{-36}$ and $\beta = .33$ (we did not optimize $\beta$ and so a better choice can lead to more optimistic settings for our bound).  With these settings, the protocol fails with probability less than $\epsilon^{1-2\beta} = 10^{-12}$ while $\epsilon_{PA} < 5\times 10^{-12}$.  A graph of the random generation rate of this protocol using our new entropic uncertainty bound is shown in Figure \ref{fig:rate}.

Note that, in the original quantum sampling paper \cite{sampling}, their method was applied to the security proof of BB84 \cite{QKD-BB84}.  However, their proof relied on many internal symmetries within BB84 which we did not need for our proof here - instead, our entropic uncertainty bound applied immediately to the QRNG protocol without requiring any additional reductions.  We believe that with further refinements to our method, along with an extension to adversaries with quantum memories, this technique of utilizing quantum sampling, \emph{augmented with the analysis framework we introduced in our proof of Theorem \ref{thm:main-qubit}}, can lead to a powerful mechanism for proving security of cryptographic protocols in finite key settings.


\begin{figure}
  \centering
  \includegraphics[width=225pt]{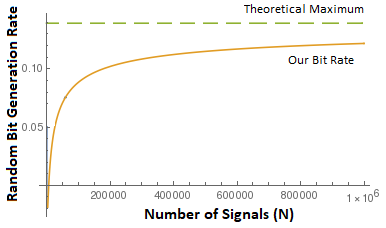}
  \caption{Showing the random bit generation rate we derived using our entropic uncertainty relation (Solid line), namely $\ell_{QRNG} / N$ as the number of signals $N = n+m$ increases.  We assume a high source noise level of $20\%$ here (namely, $w(q) = .2$).  We use $m = 0.07n$ in this graph and $\beta = .33$.  Neither settings were optimized, so the result could potentially be improved further.  Also showing the theoretical, asymptotic upper bound (dashed line) for this same noise level.  We note that, as $N$ increases beyond the plotted $10^6$, our lower-bound numerically tends to approach the theoretical maximum.}\label{fig:rate}
\end{figure}

\section{Closing Remarks}
In this paper we showed an interesting connection between quantum sampling and quantum uncertainty.  We used the quantum sampling technique introduced in \cite{sampling} to derive and prove a new entropic uncertainty relation based on smooth min entropy, the Shannon entropy of an observed outcome, and the probability of failure of a classical sampling strategy.  Our result is applicable to arbitrary, finite, states that are not necessarily i.i.d.  From this we were able to derive an alternative, and simple, proof for the Maassen and Uffink bound first proven in \cite{MU-bound}.  We also showed how our result can be used to derive bit generation rates for quantum random number generators where the source is controlled by a memory-less adversary.  To our knowledge, this is the first time quantum sampling has been extended to general quantum information theory and our method of proving Theorem \ref{thm:main-qubit} may hold broad application in future research.  Note that, though we only proved the qubit case of the Maassen and Uffink entropic uncertainty relation, we strongly suspect this technique can be used to prove the higher dimensional case also. It would also be interesting to see if quantum sampling can yield a simple proof for the conditional version of the uncertainty relation, namely $H(M|B) + H(N|E) \ge -\log c$ \cite{survey,smooth-uncertainty}.  We are currently investigating this, also, as future work.  Finally, investigating our method's application to other cryptographic protocols is another interesting line of investigation.
$ $\newline\newline
\textbf{Acknowledgments:} The author would like to thank the anonymous reviewers for their comments which have greatly improved the presentation of this paper.  The author would also like to acknowledge support from National Science Foundation grant number 1812070.


\end{document}